\documentclass{amsart}
\usepackage{amssymb}
\usepackage{amsmath}
\usepackage{amsfonts}

\newtheorem{theorem}{Theorem}
\theoremstyle{plain}

\newtheorem{proposition}{Proposition}
\newtheorem{remark}{Remark}

\numberwithin{equation}{section}
\newcommand{\TCItag}[1]{\nonumber\quad\text{(#1)}}
\newcommand{\func}[1]{\operatorname{#1}}
\newcommand{\tciFourier}{\mathcal{F}}

\begin{document}
\title[Inverse scattering method via Gelfand-Levitan-Marchenko equation]{%
Inverse scattering method for nonlinear negative first order coupled
Klein--Gordon equation }
\author{Cihan Sabaz$^{1}$, Ilmar Gahramanov$^{2,\text{ }3}$ and Mansur I.
Ismailov$^{1,\text{ }3}$, }
\email{cihansabaz@gtu.edu.tr, ilmar.gahramanov@bogazici.edu.tr,
mismailov@gtu.edu.tr }
\date{}
\dedicatory{$^{1}$\textit{\ Department of Mathematics, Gebze Technical
University, \ 41400 Gebze-Kocaeli, T\"{u}rkiye}\\
$^{2}$ \textit{Department of Physics, Bogazici University, 34342 Bebek,
Istanbul, T\"{u}rkiye}\\
$^{3}$ \textit{Center for Mathematics and its Applications, Khazar
University, 1096 Baku, Azerbaijan}}

\begin{abstract}
The exact N-soliton solutions are derived for the coupled negative first
order Klein-Gordon (CNKG) equation subject to vanishing boundary conditions
by using the inverse scattering method via Gelfand-Levitan-Marchenko
equation. Based on the zero-curvature representation, the conservation laws,
integrals of motion, and Hamiltonian structure of the aforementioned coupled
nonlinear equations are constructed. The Jost functions and their
analyticity properties for the Manakov spectral problem are recalled. The
integral equations for the eigenfunctions are then used to formulate the
Gel'fand--Levitan--Marchenko equations. Solving these equations establishes
a direct correspondence between the kernel functions and the potential,
yielding the general N-soliton expressions.
\end{abstract}

\maketitle

\section{ Introduction}

The work of Tsuchida and Wadati [1] considers multi-component versions of
the modified Korteweg--de Vries (mKdV) equation and the nonlinear Schr\"{o}%
dinger equation, generalizing the coupled mKdV system previously introduced
by Iwao and Hirota [2], who derived its multi-soliton solutions. These
coupled equations have been solved by the inverse scattering method (ISM)
through the Gelfand--Levitan--Marchenko (GLM) equation. However, it remains
unknown whether the coupled negative-order nonlinear partial differential
equations and the corresponding hierarchy can also be solved by the ISM.

Consider the following coupled negative first order Klein-Gordon (CNKG):

\begin{eqnarray}
q_{2,xt} &=&q_{2}+2q_{2}\partial ^{-1}\left( \left\vert q_{2}^{2}\right\vert
\right) _{t}+q_{2}\partial ^{-1}\left( \left\vert q_{3}^{2}\right\vert
\right) _{t}+q_{3}\partial ^{-1}\left( q_{2}q_{3}^{\ast }\right) _{t}, 
\notag \\
&&  \TCItag{1.1} \\
q_{3,xt} &=&q_{3}+2q_{3}\partial ^{-1}\left( \left\vert q_{3}^{2}\right\vert
\right) _{t}+q_{3}\partial ^{-1}\left( \left\vert q_{2}^{2}\right\vert
\right) _{t}+q_{2}\partial ^{-1}\left( q_{2}^{\ast }q_{3}\right) _{t}, 
\notag
\end{eqnarray}%
where $\partial _{x}^{-1}=\int\limits_{x}^{\infty }dx$\ is indefinite
integral with respect to $x$.

In the case $q_{2}=q_{3}=\pm \frac{1}{\sqrt{2}}q$ the system (1.1) becomes
to the equation

\begin{equation*}
q_{xt}=q+4q\partial ^{-1}(\left\vert q^{2}\right\vert )_{t}
\end{equation*}%
makes it scalar negative order nonlinear Klein-Gordon Equation (see [3,4]),
where is the nonlinear Klein--Gordon equation coupled to a scalar field in
[5]

\begin{eqnarray*}
r_{\varkappa \varkappa }-r_{\tau \tau }-r+2r^{3}+pr &=&0, \\
p_{\varkappa }-p_{\tau }-4rr_{\tau } &=&0
\end{eqnarray*}%
by the change of variables $\varkappa =\frac{x+4t}{2}$, $\tau =\frac{4t+x}{2}
$, elimination of $p$ in second equation under the assumption that the
scalar field $p$ tends to zero at infinity and by the substitution $r=\frac{1%
}{2}q$.

In the present paper, the inverse scattering method (ISM) is employed to
investigate soliton solutions of the CNKG equation. Through spectral
analysis, the time evolution of the Gelfand--Levitan--Marchenko (GLM)
equations is derived, and a correspondence between the solutions of the CNKG
equation and those of the GLM equations is established. By solving the GLM
equations, explicit expressions for the soliton solutions of the CNKG
equation are obtained. The general framework of the inverse scattering
method based on the GLM equation is presented in [6, 7].

Recent works [8, 9] on the inverse scattering method are particularly
important, as they highlight several key challenges, especially concerning
the construction and application of the inverse scattering transform based
on the GLM equation for obtaining soliton solutions, as demonstrated in [10]
in the context of the theory of self-induced transparency. Several coupled
nonlinear equations, including the mKdV equation, are reconsidered in [11]
within the framework of the inverse scattering method via the
Riemann--Hilbert problem.

The paper is organized as follows. In Section 2, the zero-curvature
representation of the CNKG equation are derived using their associated Lax
pair by the Manakov spectral problem. This section also introduces the
conservation laws, integral of motion{} and Hamiltonian structure of CNKG
equation. In Section 3, the spectral problem corresponding to the Manakov
system, presents a rigorous analysis of the Jost functions, examines the
asymptotic behavior of the scattering data, and introduces the adjoint
spectral problem. Furthermore, the fundamental properties of the scattering
matrix are discussed, and the GLM equations are derived rigorously. The
solution of the CNGK equation are then expressed in terms of the scattering
data, where the zeros of the diagonal entries of the scattering matrix
correspond to soliton solutions. In Section 4, the time evolution of the
scattering data is studied. The N-soliton solutions of the CNKG equation are
constructed explicitly via the GLM equation. Finally, the one and
two-solitons of the CNGKequation are illustrated by appropriate parameter
choice in the structural features of the N-soliton solutions.

\section{ Lax Pair}

In this section, we represent the CNKG equation as the compatibility
condition of a pair of linear equations. These may consist of an eigenvalue
problem for the Lax operator $L$ and a linear evolution equation generated
by the operator A for the corresponding eigenfunction. Alternatively, they
may take the form of a first-order linear system expressing the covariant
constancy of a vector with respect to a 1+1-dimensional gauge potential $%
(U,V)$. The compatibility condition is given either by the Lax equation $%
\frac{d}{dt}L=[L,A]$ or equivalently by the zero-curvature condition $%
U_{t}-V_{x}+[U,V]=0.$ The conserved quantities then follow from the
isospectrality of the Lax operator and the associated monodromy matrix.

\textbf{2.1. Zero-curvature representation of CNKG equation:}

Consider the Manakov spectral problem ([12]) of 3 $\times $ 3 linear system

\begin{equation}
\left[ 
\begin{array}{c}
\Psi _{1x} \\ 
\Psi _{2x} \\ 
\Psi _{3x}%
\end{array}%
\right] =U(q_{2},q_{3})\left[ 
\begin{array}{c}
\Psi _{1} \\ 
\Psi _{2} \\ 
\Psi _{3}%
\end{array}%
\right] ,  \tag{2.1}
\end{equation}%
where $U(q_{2},q_{3})=\left[ 
\begin{array}{ccc}
-i\mu & q_{2} & q_{3} \\ 
-q_{2}^{\ast } & i\mu & 0 \\ 
-q_{3}^{\ast } & 0 & i\mu%
\end{array}%
\right] $ with $\mu $ is a nonzero eigenvalue, $\Psi _{1},\Psi _{2}$ and $%
\Psi _{3}$ are linearly independent eigenfunctions, $i^{2}=-1;$ $%
q_{2}=q_{2}(x,t)$ and $q_{3}=q_{3}(x,t)$ are the rapidly decreasing at
infinity complex valued coefficients.The auxiliary spectral problem
described as follows:%
\begin{equation}
\left[ 
\begin{array}{c}
\Psi _{1t} \\ 
\Psi _{2t} \\ 
\Psi _{3t}%
\end{array}%
\right] =V(q_{2},q_{3})\left[ 
\begin{array}{c}
\Psi _{1} \\ 
\Psi _{2} \\ 
\Psi _{3}%
\end{array}%
\right] ,  \tag{2.2}
\end{equation}%
where $V(q_{2},q_{3})=\left[ 
\begin{array}{ccc}
a & b & c \\ 
d & e & f \\ 
k & l & m%
\end{array}%
\right] $ and $a,b,c,d,e,f,k,l$ and $m$ are scalar functions, independent of 
$\Psi _{1},\Psi _{2}$ and $\Psi _{3}.$

\bigskip From (2.1) and (2.2), the zero curvature equation $U_{t}-V_{x}+%
\left[ U,V\right] =0$ yields

\begin{equation*}
\begin{array}{c}
a_{x}=q_{2}d+q_{3}k+q_{2}^{\ast }b+q_{3}^{\ast }c, \\ 
b_{x}=-2i\mu b-q_{2}\alpha +q_{2}e+q_{3}l+q_{2,t}, \\ 
c_{x}=-2i\mu c-q_{3}\alpha +q_{2}f+q_{3}m+q_{3,t},%
\end{array}%
\begin{array}{c}
f_{x}=-q_{2}^{\ast }c-q_{3}d, \\ 
l_{x}=-q_{3}^{\ast }b-q_{2}k, \\ 
m_{x}=-q_{3}^{\ast }c-q_{3}k,%
\end{array}%
\end{equation*}%
\begin{equation}
\tag{2.3}
\end{equation}%
\begin{equation*}
\begin{array}{c}
d_{x}=2i\mu d-q_{2}^{\ast }\alpha +q_{2}^{\ast }e+q_{3}^{\ast
}f-q_{2,t}^{\ast }, \\ 
k_{x}=2i\mu k+q_{2}^{\ast }l+q_{3}^{\ast }m-q_{3}^{\ast }a-q_{3,t}^{\ast },
\\ 
e_{x}=-q_{2}^{\ast }b-q_{2}d,%
\end{array}%
\end{equation*}%
Let the following transformations be applied to the system (2.3):

\begin{eqnarray*}
a &=&\frac{A(x,t)}{\mu },b=\frac{B(x,t)}{\mu },c=\frac{C(x,t)}{\mu }, \\
d &=&\frac{D(x,t)}{\mu },e=\frac{E(x,t)}{\mu },f=\frac{F(x,t)}{\mu }, \\
k &=&\frac{K(x,t)}{\mu },l=\frac{L(x,t)}{\mu },m=\frac{M(x,t)}{\mu }.
\end{eqnarray*}%
As a result the following equations are obtained:

\begin{equation*}
\begin{array}{cc}
A_{x}=q_{2}D+q_{3}K+q_{2}^{\ast }B+q_{3}^{\ast }C, & E_{x}=-q_{2}^{\ast
}B-q_{2}D, \\ 
B_{x}=q_{2}E-q_{2}A+q_{3}L,\text{ }B=-\frac{i}{2}q_{2,t}, & 
L_{x}=-q_{3}^{\ast }B-q_{2}K, \\ 
C_{x}=q_{2}F+q_{3}M-q_{3}A,\text{ }C=-\frac{i}{2}q_{3,t}, & 
M_{x}=-q_{3}^{\ast }C-q_{3}K,%
\end{array}%
\end{equation*}%
\begin{equation}
\tag{2.4}
\end{equation}%
\begin{equation*}
\begin{array}{c}
D_{x}=-q_{2}^{\ast }A+q_{2}^{\ast }E+q_{3}^{\ast }F,\text{ }D=-\frac{i}{2}%
q_{2,t}^{\ast }, \\ 
K_{x}=q_{2}^{\ast }L-q_{3}^{\ast }A+q_{3}^{\ast }M,\text{ }K=-\frac{i}{2}%
q_{3,t}^{\ast }, \\ 
F_{x}=-q_{2}^{\ast }C-q_{3}D,%
\end{array}%
\end{equation*}

The following negative first order AKNS equation is obtained for important
case of spectral problem (2.1).

\begin{proposition}
Let $q_{2}$ and $q_{3}$ be the coefficients of the system (2.1), then the
system of equations (2.4) becomes the following negative order pair of
equations:%
\begin{eqnarray}
q_{2,xt} &=&q_{2}\left[ 2\partial ^{-1}\left( \left\vert
q_{2}^{2}\right\vert \right) _{t}+\partial ^{-1}\left( \left\vert
q_{3}^{2}\right\vert \right) _{t}+\gamma \right] +q_{3}\partial ^{-1}\left(
q_{2}q_{3}^{\ast }\right) _{t},  \TCItag{2.5} \\
q_{3,xt} &=&q_{2}\partial ^{-1}\left( q_{2}^{\ast }q_{3}\right) _{t}+q_{3} 
\left[ \partial ^{-1}\left( \left\vert q_{2}^{2}\right\vert \right)
_{t}+2\partial ^{-1}\left( \left\vert q_{3}^{2}\right\vert \right)
_{t}+\gamma \right] ,  \notag
\end{eqnarray}%
where $\gamma $ is a arbitrary non-zero real constant and $\partial
_{x}^{-1}=\int\limits_{x}^{\infty }dx$\ is indefinite integral with respect
to $x$.

\begin{proof}
It is clearly seen that $A_{x}=-E_{x}-M_{x},$ $B=-D^{\ast },$ $C=-K^{\ast }$
and $F_{x}=-L_{x}^{\ast }$ in the system (2.4). This system becomes

\begin{equation*}
B=-\frac{i}{2}q_{2,t},\text{ }C=-\frac{i}{2}q_{3,t},
\end{equation*}

\begin{equation*}
E_{x}=\frac{i}{2}\left( \left \vert q_{2}^{2}\right \vert \right) _{t},\text{
}F_{x}=\frac{i}{2}\left( q_{2}^{\ast }q_{3}\right) _{t},\text{ }M_{x}=\frac{i%
}{2}\left( \left \vert q_{3}^{2}\right \vert \right) _{t},
\end{equation*}

\begin{equation*}
B_{x}=q_{2}\left( 2E+M\right) -q_{3}F^{\ast },
\end{equation*}

\begin{equation*}
C_{x}=q_{2}F+q_{3}\left( E+2M\right) .
\end{equation*}

In this system, if some terms are integrated from $x$ to $\infty ,$

\begin{eqnarray*}
E &=&-\frac{i}{2}\partial ^{-1}\left( \left \vert e_{2}^{2}\right \vert
\right) _{t}-\frac{i}{2}c_{1}, \\
F &=&-\frac{i}{2}\partial ^{-1}\left( e_{2}^{\ast }e_{3}\right) _{t}, \\
M &=&-\frac{i}{2}\partial ^{-1}\left( \left \vert e_{3}^{2}\right \vert
\right) _{t}-\frac{i}{2}c_{1}, \\
A &=&\frac{i}{2}\partial ^{-1}\left( \left \vert e_{2}^{2}\right \vert
\right) _{t}+\frac{i}{2}\partial ^{-1}\left( \left \vert e_{3}^{2}\right
\vert \right) _{t}-\frac{i}{2}c_{2},
\end{eqnarray*}

are obtained. Here, $c_{1}$ and $c_{2}$ are arbitrary real constants and $%
\partial _{x}^{-1}=\int\limits_{x}^{\infty }dx$\ is indefinite integral with
respect to $x$. Consequently, the matrix $V(q_{2},q_{3})$ in (2.2) is given
by 
\begin{equation*}
V(q_{2},q_{3})=\frac{1}{\mu }\left[ 
\begin{array}{ccc}
\frac{i}{2}\partial ^{-1}\left( \left\vert e_{2}^{2}\right\vert \right) _{t}+%
\frac{i}{2}\partial ^{-1}\left( \left\vert e_{3}^{2}\right\vert \right) _{t}-%
\frac{i}{2}c_{2} & -\frac{i}{2}q_{2,t} & -\frac{i}{2}q_{3,t} \\ 
-\frac{i}{2}q_{2,t}^{\ast } & -\frac{i}{2}\partial ^{-1}\left( \left\vert
e_{2}^{2}\right\vert \right) _{t}-\frac{i}{2}c_{1} & -\frac{i}{2}\partial
^{-1}\left( e_{2}^{\ast }e_{3}\right) _{t} \\ 
-\frac{i}{2}q_{3,t}^{\ast } & -\frac{i}{2}\partial ^{-1}\left( e_{3}^{\ast
}e_{2}\right) _{t} & -\frac{i}{2}\partial ^{-1}\left( \left\vert
e_{3}^{2}\right\vert \right) _{t}-\frac{i}{2}c_{1}%
\end{array}%
\right] .
\end{equation*}

For the compatibility of equations (2.1) and (2.2) the functions $q_{2}$ and 
$q_{3}$ must satisfy the system (2.5), where $\gamma =c_{1}-c_{2},$ $%
(c_{1}\neq c_{2}).$
\end{proof}
\end{proposition}

\begin{remark}
The case $q_{2}=q_{3}=\pm q$ becomes to the equation

\begin{equation*}
q_{xt}=q[2\partial ^{-1}(\left\vert q^{2}\right\vert )_{t}+\partial
^{-1}(\left\vert q^{2}\right\vert )_{t}+\varkappa ]+q\partial
^{-1}(\left\vert q^{2}\right\vert )_{t}
\end{equation*}

\begin{equation*}
\Longrightarrow q_{xt}=4q\partial ^{-1}(\left\vert q^{2}\right\vert
)_{t}+\gamma q
\end{equation*}

and the change $q\rightarrow \frac{1}{\sqrt{2}}q$ with $\gamma =1$ makes it
scalar negative order nonlinear Klein-Gordon Equation (see [3,4]).
\end{remark}

\textbf{2.2. Conservation laws{}:}

Introducing

\begin{equation*}
E=\left( 
\begin{array}{c}
q_{2} \\ 
q_{3}%
\end{array}%
\right) ,\text{ \ \ }\rho =EE^{\dagger }=\left( 
\begin{array}{cc}
\left\vert q_{2}\right\vert ^{2} & q_{2}q_{3}^{\ast } \\ 
q_{2}^{\ast }q_{3} & \left\vert q_{3}\right\vert ^{2}%
\end{array}%
\right)
\end{equation*}

the system (2.5) becomes

\begin{equation*}
E_{xt}=\left[ \partial _{x}^{-1}(\rho _{t})+\partial _{x}^{-1}(Tr\rho
)_{t}I+\gamma I\right] E.
\end{equation*}

Defining

\begin{equation*}
R=\partial _{x}^{-1}(\rho _{t}),
\end{equation*}

we obtain

\begin{equation*}
E_{xt}=\left[ R+(TrR)I+\gamma I\right] E.
\end{equation*}

This form makes the global $SU(2)$ symmetry manifest. Under a constant
transformation

\begin{equation*}
E\rightarrow UE,\text{ \ \ }U\in SU(2),
\end{equation*}

we have

\begin{equation*}
\rho \rightarrow U\rho U^{\dagger },\text{ \ \ }R\rightarrow URU^{\dagger },%
\text{ \ \ }TrR\rightarrow TrR.
\end{equation*}

Hence the equation is $SU(2)-$covariant.

Now use the Cayley-Hamilton theorem. Since

\begin{equation*}
\rho =EE^{\dagger },
\end{equation*}

we have

\begin{equation*}
\det \rho =0.
\end{equation*}

Therefore

\begin{equation*}
\rho ^{2}-(Tr\rho )\rho +(\det \rho )I=0
\end{equation*}

reduces to

\begin{equation*}
\rho ^{2}=(Tr\rho )\rho .
\end{equation*}

Thus $\rho $ has rank at most one (rank one whenever $E\neq 0$).

For the integrated matrix $R,$ however, there is no reason to assume $\det
R=0.$ Instead, decompose $R$ into trace and traceless parts:

\begin{equation*}
R=Q+\frac{1}{2}(TrR)I,\text{ \ \ }TrQ=0.
\end{equation*}

Since $R$ is Hermitian, $Q$ is traceless Hermitian matrix, i.e. $Q\in
isu(2). $

Substituting into the equation gives

\begin{equation*}
E_{xt}=\left[ Q+\frac{3}{2}(TrR)I+\gamma I\right] E.
\end{equation*}

Equivalently,

\begin{equation*}
E_{xt}=\left[ Q+\frac{3}{2}\partial _{x}^{-1}(Tr\rho )_{t}I+\gamma I\right]
E.
\end{equation*}

Assuming $\gamma $ is real and the integration constants entering $\partial
_{x}^{-1}$ preserve Hermiticity, the matrix multiplying $E$ is Hermitian.
Hence

\begin{equation*}
E^{\dagger }E_{xt}-E_{xt}^{\dagger }E=0.
\end{equation*}

Using

\begin{equation*}
E^{\dagger }E_{xt}=\partial _{x}(E^{\dagger }E_{t})-E_{x}^{\dagger }E_{t},
\end{equation*}

and

\begin{equation*}
E_{xt}^{\dagger }E=\partial _{t}(E_{x}^{\dagger }E)-E_{x}^{\dagger }E_{t},
\end{equation*}

we obtain the conservation law

\begin{equation*}
\partial _{x}(E^{\dagger }E_{t})-\partial _{t}(E_{x}^{\dagger }E)=0.
\end{equation*}

Taking the imaginary part yields the real conservation law

\begin{equation*}
\partial _{x}\func{Im}(E^{\dagger }E_{t})-\partial _{t}\func{Im}%
(E_{x}^{\dagger }E)=0.
\end{equation*}

Therefore, for periodic or rapidly decaying boundary conditions,

\begin{equation*}
\frac{d}{dt}\int \func{Im}(E_{x}^{\dagger }E)dx=0.
\end{equation*}

\textbf{2.3 Integrals of motion:}

The Manakov spectral problem (2.1) admits an infinite hierarchy of integrals
of motion, obtainable either from the Hermiticity argument of subsection 2.2
or from the asymptotic expansion of the associated Riccati equation. The
first five members are given below, first in vector/matrix form and then in
explicit component form.

The integrals of motion are

\bigskip

\begin{equation}
I_{1}=\int E^{\dagger }Edx,  \tag{2.6}
\end{equation}

\begin{equation}
I_{2}=\int \func{Im}(E_{x}^{\dagger }E)dx,  \tag{2.7}
\end{equation}

\begin{equation}
I_{3}=\int \left[ (E^{\dagger }E)^{2}-E_{x}^{\dagger }E_{x}\right] dx, 
\tag{2.8}
\end{equation}

\begin{equation}
I_{4}=3\int (E^{\dagger }E)\func{Im}(E_{x}^{\dagger }E)dx-\int \func{Im}%
(E_{xx}^{\dagger }E_{x})dx,  \tag{2.9}
\end{equation}

\begin{equation}
I_{5}=\int \left[ E_{xx}^{\dagger }E_{xx}-4(E^{\dagger }E)(E_{x}^{\dagger
}E_{x})-\frac{3}{2}(E_{x}^{\dagger }E+E^{\dagger }E_{x})^{2}-2\left[ \func{Im%
}(E_{x}^{\dagger }E)\right] ^{2}+2(E^{\dagger }E)^{3}\right] dx.  \tag{2.10}
\end{equation}

\bigskip

Each $I_{n}$ is manifestly $SU(2)-$invariant, since $E^{\dagger }E$, $%
E_{x}^{\dagger }E$, $E_{x}^{\dagger }E_{x}$, $E_{xx}^{\dagger }E_{x}$, $%
E_{xx}^{\dagger }E_{xx}$ are all invariant under the constant transformation 
$E\rightarrow UE$, $U\in SU(2)$ of Section 2.2.

Writing out $E^{\dagger }E=\left\vert q_{2}\right\vert ^{2}+\left\vert
q_{3}\right\vert ^{2}$, etc., the same five integrals read

\begin{equation}
I_{1}=\int (\left\vert q_{2}\right\vert ^{2}+\left\vert q_{3}\right\vert
^{2})dx,  \tag{2.6$^{\shortmid }$}
\end{equation}

\begin{equation}
I_{2}=\frac{1}{2i}\int (q_{2,x}^{\ast }q_{2}-q_{2}^{\ast
}q_{2,x}+q_{3,x}^{\ast }q_{3}-q_{3}^{\ast }q_{3,x})dx, 
\tag{2.7$^{\shortmid
}$}
\end{equation}

\begin{equation}
I_{3}=\int \left[ (\left\vert q_{2}\right\vert ^{2}+\left\vert
q_{3}\right\vert ^{2})^{2}-\left\vert q_{2,x}\right\vert ^{2}-\left\vert
q_{3,x}\right\vert ^{2}\right] dx,  \tag{2.8$^{\shortmid }$}
\end{equation}

\begin{eqnarray}
I_{4} &=&\frac{3}{2i}\int (\left\vert q_{2}\right\vert ^{2}+\left\vert
q_{3}\right\vert ^{2})(q_{2,x}^{\ast }q_{2}-q_{2}^{\ast
}q_{2,x}+q_{3,x}^{\ast }q_{3}-q_{3}^{\ast }q_{3,x})dx 
\TCItag{2.9$^{\shortmid }$} \\
&&-\frac{1}{2i}\int (q_{2,xx}^{\ast }q_{2,x}-q_{2,xx}q_{2,x}^{\ast
}+q_{3,xx}^{\ast }q_{3,x}-q_{3,xx}q_{3,x}^{\ast })dx,  \notag
\end{eqnarray}

\begin{eqnarray}
I_{5} &=&\int \{\left\vert q_{2,xx}\right\vert ^{2}+\left\vert
q_{3,xx}\right\vert ^{2}-4(\left\vert q_{2}\right\vert ^{2}+\left\vert
q_{3}\right\vert ^{2})(\left\vert q_{2,x}\right\vert ^{2}+\left\vert
q_{3,x}\right\vert ^{2})  \TCItag{2.10$^{\shortmid }$} \\
&&-\frac{3}{2}(q_{2,x}^{\ast }q_{2}+q_{2}^{\ast }q_{2,x}+q_{3,x}^{\ast
}q_{3}+q_{3}^{\ast }q_{3,x})^{2}  \notag \\
&&+\frac{1}{2}(q_{2,x}^{\ast }q_{2}-q_{2}^{\ast }q_{2,x}+q_{3,x}^{\ast
}q_{3}-q_{3}^{\ast }q_{3,x})^{2}  \notag \\
&&+2(\left\vert q_{2}\right\vert ^{2}+\left\vert q_{3}\right\vert
^{2})^{3}\}dx.  \notag
\end{eqnarray}

\bigskip

\textbf{2.4. Hamiltonian structure and negative-order flow:}

The Manakov spectral problem (2.1) also carries the standard canonical
Poisson bracket

\begin{equation}
\left\{ q_{i}(x),q_{j}^{\ast }(y)\right\} =i\delta _{ij}\delta (x-y),\text{ }%
i,j=2,3,  \tag{2.11}
\end{equation}

with respect to which

\begin{equation}
\frac{\delta I_{1}}{\delta q_{2}^{\ast }}=q_{2},\text{ \ \ }\frac{\delta
I_{1}}{\delta q_{3}^{\ast }}=q_{3},\text{ \ \ }i.e.\text{ }\frac{\delta I_{1}%
}{\delta E^{\dagger }}=E.  \tag{2.12}
\end{equation}

Under (2.11), $I_{1}$ generates the trivial flow $q_{i,t}=iq_{i}$, the CNKG
equation (2.5), being of negative order, is instead generated by $I_{1}$
through the nonlocal Hamiltonian operator

\begin{equation}
E_{t}=-\partial _{x}^{-1}[(R+(TrR)I+\gamma I)\frac{\delta I_{1}}{\delta
E^{\dagger }}].  \tag{2.13}
\end{equation}

\section{Manakov Spectral Problem}

As noted in the previous section, the principal advantage of representing
nonlinear equations through the zero-curvature condition lies in the
possibility of finding the explicit solutions by means of the auxiliary
linear equations (2.1)--(2.2). One of these equations provides a
transformation from the functions of interest to the scattering data
associated with the operator $U$. Moreover, the straightforward application
of the inverse scattering transform method requires only the
fundamentalconcepts of scattering theory. Therefore, we briefly review the
essential aspects of the Manakov problem described by equation (2.1).

Let us consider the equations of Manakov's spectral problem in the form

\begin{eqnarray}
\frac{d\psi _{1}}{dx}+i\mu \psi _{1} &=&q_{2}\psi _{2}+q_{3}\psi _{3}, 
\TCItag{3.1} \\
\frac{d\psi _{2}}{dx}-i\mu \psi _{2} &=&-q_{2}^{\ast }\psi _{1},  \notag \\
\frac{d\psi _{3}}{dx}-i\mu \psi _{3} &=&-q_{3}^{\ast }\psi _{1},  \notag
\end{eqnarray}%
where the\textquotedblright potentials\textquotedblright \ $q_{2,3}$ are
taken at some point $t$ (for example, $t=0$) and being functions of $x$. The
major stages of the analysis of a spectral problem are as follows:

\textbf{3.1. Jost Functions: }

First,we note a general property of the solutions to equation (3.1). Let $%
\psi =(\psi _{1},\psi _{2},\psi _{3})^{T}$ and $\varphi =(\varphi
_{1},\varphi _{2},\varphi _{3})^{T}$ be two solutions of equation (3.1)
corresponding to the same real value of $\mu $.Then,

\begin{equation*}
\frac{d}{dx}((\psi ^{T})^{\ast }\varphi )=0.
\end{equation*}

\bigskip

Here, $T$ denotes the transpose. Let $R$ denote the matrix constructed from
the column vectors of the solutions of equation (3.1). Then,

\begin{equation*}
\frac{d}{dx}\det R=i\mu \det R.
\end{equation*}

\bigskip

If the potentials $q_{2,3}$ decay sufficiently rapidly as $x\rightarrow \pm
\infty ,$ the solution of equation (3.1) is uniquely determined by either of
its asymptotic behaviours in the limits $x\rightarrow \pm \infty $. We
therefore consider two sets of solutions to equation (3.1), denoted by $\Phi
^{(i)}(x,\mu )$ and $\Psi ^{(i)}(x,\mu ),$ $i=1,2,3.$ These solutions,
commonly referred to as the Jost functions, are characterized by the
following asymptotic forms for real values of $\mu :$

\bigskip

\begin{eqnarray*}
\Phi ^{(1)} &\rightarrow &\left( 
\begin{array}{c}
1 \\ 
0 \\ 
0%
\end{array}%
\right) e^{-i\mu x},\text{ \ \ }\Phi ^{(2)}\rightarrow \left( 
\begin{array}{c}
0 \\ 
1 \\ 
0%
\end{array}%
\right) e^{i\mu x},\text{ \ \ }\Phi ^{(3)}\rightarrow \left( 
\begin{array}{c}
0 \\ 
0 \\ 
1%
\end{array}%
\right) e^{i\mu x},\text{ \ \ }x\rightarrow -\infty , \\
\Psi ^{(1)} &\rightarrow &\left( 
\begin{array}{c}
1 \\ 
0 \\ 
0%
\end{array}%
\right) e^{-i\mu x},\text{ \ \ }\Psi ^{(2)}\rightarrow \left( 
\begin{array}{c}
0 \\ 
1 \\ 
0%
\end{array}%
\right) e^{i\mu x},\text{ \ \ }\Psi ^{(3)}\rightarrow \left( 
\begin{array}{c}
0 \\ 
0 \\ 
1%
\end{array}%
\right) e^{i\mu x},\text{ \ \ }x\rightarrow +\infty .
\end{eqnarray*}

\bigskip

The Jost functions constitute a complete set of linearly independent
solutions to equation (3.1) and, therefore, can be represented in terms of
one another as follows:

\begin{equation}
\Phi ^{(i)}(x,\mu )=\underset{j=1}{\overset{3}{\sum }}S_{ij}(\mu )\Psi
^{(j)}(x,\mu ),  \tag{3.2}
\end{equation}%
where $S_{ij}$ denotes the transition, or scattering matrix. The following
relations follow directly from the properties established above:

\bigskip

\begin{eqnarray*}
(\Phi ^{(i)T})^{\ast }(x,\mu )\Phi ^{(j)}(x,\mu ) &=&\delta _{ij}, \\
(\Psi ^{(i)T})^{\ast }(x,\mu )\Psi ^{(j)}(x,\mu ) &=&\delta _{ij}, \\
(\Psi ^{(j)T})^{\ast }(x,\mu )\Phi ^{(i)}(x,\mu ) &=&S_{ij}(\mu ).
\end{eqnarray*}

\bigskip

Here, $T$ denotes the transpose. In a scattering problem, the analyticity
properties of the Jost functions and the transition matrix play a
fundamental role.

\textbf{3.2 Analytical Properties:}

To investigate the analyticity properties of the Jost functions, it is
convenient to

rewrite equation (3.1) in its integral form for $i=1,2,3:$

\bigskip

\begin{equation*}
\Phi ^{(i)}(x,\mu )=F^{(i)}(x,\mu )+\overset{x}{\underset{-\infty }{\int }}%
e^{-i\mu I(x-x^{\prime })}E(x^{\prime })\Phi ^{(i)}(x^{\prime })dx^{\prime },
\end{equation*}

\begin{equation*}
\Psi ^{(i)}(x,\mu )=F^{(i)}(x,\mu )-\overset{x}{\underset{-\infty }{\int }}%
e^{-i\mu I(x-x^{\prime })}E(x^{\prime })\Psi ^{(i)}(x^{\prime })dx^{\prime },
\end{equation*}

where

\bigskip

\begin{eqnarray*}
F^{(1)}(x,\mu ) &=&h^{(1)}e^{-i\mu x},\text{ \ \ }F^{(2,3)}(x,\mu
)=h^{(2,3)}e^{i\mu x}, \\
h^{(1)} &=&\left( 
\begin{array}{c}
1 \\ 
0 \\ 
0%
\end{array}%
\right) ,\text{ \ \ }h^{(2)}=\left( 
\begin{array}{c}
0 \\ 
1 \\ 
0%
\end{array}%
\right) ,\text{ \ \ }h^{(3)}=\left( 
\begin{array}{c}
0 \\ 
0 \\ 
1%
\end{array}%
\right) , \\
I &=&\left( 
\begin{array}{ccc}
1 & 0 & 0 \\ 
0 & -1 & 0 \\ 
0 & 0 & -1%
\end{array}%
\right) ,\text{ \ \ }E=\left( 
\begin{array}{ccc}
0 & q_{2} & q_{3} \\ 
-q_{2}^{\ast } & 0 & 0 \\ 
-q_{3}^{\ast } & 0 & 0%
\end{array}%
\right) =-(E^{T})^{\ast }.
\end{eqnarray*}

\bigskip

Here, $T$ denotes the transpose. To examine the analyticity properties of $%
\Phi ^{(1)},$ it is natural to introduce the following functions:

\bigskip

\begin{equation*}
\chi _{i}(x,\mu )=\varphi _{i}^{(1)}(x,\mu )e^{i\mu x},\text{ }i=1,2,3,
\end{equation*}%
for which the following relations hold:

\begin{eqnarray*}
\chi _{1}(x,\mu ) &=&1-\underset{-\infty }{\overset{x}{\int }}dx^{\prime }%
\underset{-\infty }{\overset{x^{\prime }}{\int }}dx^{\prime \prime }e^{2i\mu
(x^{\prime }-x^{\prime \prime })}\chi _{1}(x^{\prime \prime },\mu
)[q_{2}(x^{\prime })q_{2}^{\ast }(x^{\prime \prime })+q_{3}(x^{\prime
})q_{3}^{\ast }(x^{\prime \prime })], \\
\chi _{2}(x,\mu ) &=&-\underset{-\infty }{\overset{x}{\int }}e^{-2i\mu
(x^{\prime }-x)}q_{2}^{\ast }(x^{\prime })\chi _{1}(x^{\prime },\mu
)dx^{\prime }, \\
\chi _{3}(x,\mu ) &=&-\underset{-\infty }{\overset{x}{\int }}e^{-2i\mu
(x^{\prime }-x)}q_{3}^{\ast }(x^{\prime })\chi _{3}(x^{\prime },\mu
)dx^{\prime }.
\end{eqnarray*}

\bigskip

It follows immediately that $\chi _{1}(x,\mu )$ admits an analytic
continuation into the upper half-plane of the complex variable $\mu $.
Furthermore, $\chi _{1}(x,\mu )\rightarrow 1$ as $\left \vert \mu
\right
\vert \rightarrow \infty $. The functions $\chi _{2,3}(x,\mu )$ are
likewise analytic in the same domain and satisfy $\chi _{2,3}(x,\mu
)\rightarrow 0$ as $\left \vert \mu \right \vert \rightarrow \infty $.
Consequently, $\Phi ^{(1)}(x,\mu )$ also admits an analytic continuation
into the upper half-plane of $\mu $. The analyticity properties of the
remaining Jost functions can be established by an analogous argument.

In summary, the functions $\Phi ^{(1)},$ $\Psi ^{(2)},$ $\Psi ^{(3)},$ $%
(\Psi ^{(1)T})^{\ast },$ $(\Phi ^{(2)T})^{\ast },$ $(\Phi ^{(3)T})^{\ast }$
are analytic in the region $\func{Im}\mu \geq 0$, whereas $\Psi ^{(1)},$ $%
\Phi ^{(2)},$ $\Phi ^{(3)},$ $(\Phi ^{(1)T})^{\ast },$ $(\Psi ^{(2)T})^{\ast
},$ $(\Psi ^{(3)T})^{\ast }$ are analytic in the region $\func{Im}\mu \leq 0$%
. Here, $T$ denotes the transpose.

The analyticity properties make it possible to construct the so-called
triangular representations. As an illustrative example, let us consider $%
\Phi ^{(1)}(x,\mu )$. Introducing the functions $\chi _{1,2,3}(x,\mu )$, we
then perform a Fourier transformation:

\bigskip

\begin{eqnarray*}
X_{1}(x,\xi ) &=&\frac{1}{2\pi }\overset{\infty }{\underset{-\infty }{\int }}%
[\chi _{1}(x,\mu )-1]e^{i\lambda \xi }d\lambda , \\
X_{2,3}(x,\xi ) &=&\frac{1}{2\pi }\overset{\infty }{\underset{-\infty }{\int 
}}\chi _{2,3}(x,\mu )e^{i\lambda \xi }d\lambda .
\end{eqnarray*}

\bigskip

Using the analyticity properties and closing the contour of integration in
the upper half-plane, we obtain

\begin{equation*}
X_{i}(x,\xi )=0\text{ for }\xi >0\text{.}
\end{equation*}

Thus

\begin{eqnarray*}
\chi _{1}(x,\mu ) &=&1+\overset{0}{\underset{-\infty }{\int }}X_{1}(x,\xi
)e^{-i\lambda \xi }d\xi , \\
\chi _{2,3}(x,\mu ) &=&\overset{0}{\underset{-\infty }{\int }}X_{2,3}(x,\xi
)e^{-i\lambda \xi }d\xi .
\end{eqnarray*}

After a simple change of variables in the integral, it is straightforward to
show that

\begin{equation*}
\Phi ^{(1)}(x,\mu )=F^{(1)}(x,\mu )+\overset{x}{\underset{-\infty }{\int }}%
N^{(1)}(x,\xi )e^{-i\lambda \xi }d\xi .
\end{equation*}

By an analogous argument, one obtains

\begin{eqnarray*}
\Phi ^{(2,3)}(x,\mu ) &=&F^{(2,3)}(x,\mu )+\overset{x}{\underset{-\infty }{%
\int }}N^{(2,3)}(x,\xi )e^{i\lambda \xi }d\xi , \\
\Psi ^{(1)}(x,\mu ) &=&F^{(1)}(x,\mu )+\overset{\infty }{\underset{x}{\int }}%
K^{(1)}(x,\xi )e^{-i\lambda \xi }d\xi , \\
\Psi ^{(2,3)}(x,\mu ) &=&F^{(2,3)}(x,\mu )+\overset{\infty }{\underset{x}{%
\int }}K^{(2,3)}(x,\xi )e^{i\lambda \xi }d\xi ,
\end{eqnarray*}%
where the column vectors $N^{(i)}(x,\xi )$ and $K^{(i)}(x,\xi ),i=1,2,3,$
denote the kernels of the triangular representations and are independent of
the parameter $\mu $.

We now turn to the analyticity properties of the transition matrix. It
follows that $S_{11}(\mu )=\Psi ^{(1)+}(x,\mu )\Phi ^{(1)}(x,\mu )$ is
analytic in the region $\func{Im}\mu \geq 0,$ whereas $S_{22}(\mu ),$ $%
S_{23}(\mu ),$ $S_{32}(\mu )$ and $S_{33}(\mu )$ are analytic in the region $%
\func{Im}\mu \leq 0.$ Furthermore, $S_{11}(\mu )\rightarrow 1$ as $%
\left
\vert \mu \right \vert \rightarrow \infty .$ In addition, by
definition, the scattering matrix is both unitary, $(S^{T})^{\ast }=S,$ and
unimodular, $\det S=1.$

\textbf{3.3. Scattering Data: }

We now consider the bounded solutions of equation (3.1) that vanish as $%
x\rightarrow \pm \infty .$ Let $\mu =i\zeta ,$ $\zeta >0.$ Then,

\begin{eqnarray*}
\widetilde{\Phi }_{k}^{(i)} &\rightarrow &e^{\zeta x}\left( 
\begin{array}{ccc}
1 & 0 & 0 \\ 
0 & 0 & 0 \\ 
0 & 0 & 0%
\end{array}%
\right) ,\text{ for }x\rightarrow -\infty , \\
\widetilde{\Psi }_{k}^{(i)} &\rightarrow &e^{-\zeta x}\left( 
\begin{array}{ccc}
1 & 0 & 0 \\ 
0 & 0 & 0 \\ 
0 & 0 & 0%
\end{array}%
\right) ,\text{ for }x\rightarrow +\infty .
\end{eqnarray*}

Therefore $\widetilde{\Phi }^{(1)}=C_{12}\widetilde{\Psi }^{(2)}+C_{13}%
\widetilde{\Psi }^{(3)}$. Since $\Phi ^{(1)}$ and $S_{11}$ admit analytic
continuation into the upper half-plane, the bounded solutions described
above correspond to the zeros of $S_{11},$ namely, $S_{11}(i\zeta )=0.$

The quantities $C_{12},$ $C_{13},$ $S_{ij}(\mu ),$ and the zeros of $%
S_{11}(\mu )$ constitute a set of scattering data that is uniquely
determined by the potentials $q_{2,3}(x).$ These data provide a direct
characterization of the scattering problem.

\textbf{3.4. GLM Equation: }

We now show that the scattering data can be used to reconstruct the
potentials $q_{2,3}(x)$ and to derive the Gel'fand--Levitan--Marchenko
equation, thereby providing a solution to the inverse scattering problem. We
begin with the first equation of (3.2),

\begin{equation*}
\Phi ^{(1)}(x,\mu )=S_{11}(\mu )\Psi ^{(1)}(x,\mu )+S_{12}(\mu )\Psi
^{(2)}(x,\mu )+S_{13}(\mu )\Psi ^{(3)}(x,\mu ),
\end{equation*}%
can be transformed, $\xi >x$, into the form

\begin{eqnarray}
&&\frac{1}{2\pi }\overset{\infty }{\underset{-\infty }{\int }}[\frac{\Phi
^{(1)}(x,\mu )}{S_{11}(\mu )}-\left( 
\begin{array}{c}
1 \\ 
0 \\ 
0%
\end{array}%
\right) e^{-i\mu x}]e^{i\lambda \xi }d\mu  \TCItag{3.3} \\
&=&\frac{1}{2\pi }\overset{\infty }{\underset{-\infty }{\int }}[\Psi
^{(1)}(x,\mu )-\left( 
\begin{array}{c}
1 \\ 
0 \\ 
0%
\end{array}%
\right) e^{-i\mu x}]e^{i\lambda \xi }d\mu  \notag \\
&&+\frac{1}{2\pi }\overset{\infty }{\underset{-\infty }{\int }}[\frac{%
S_{12}(\mu )}{S_{11}(\mu )}\Psi ^{(2)}(x,\mu )e^{i\lambda \xi }d\mu +\frac{1%
}{2\pi }\overset{\infty }{\underset{-\infty }{\int }}[\frac{S_{13}(\mu )}{%
S_{11}(\mu )}\Psi ^{(3)}(x,\mu )e^{i\lambda \xi }d\mu .  \notag
\end{eqnarray}

On the left-hand side of this equation, the integrand is analytic through
out the complex $\mu $-plane, except at the zeros of $S_{11}(\mu )\neq 0.$
Let these zeros be denoted by $\mu _{n}$, where $\func{Im}\mu _{n}>0$, $%
n=1,...,N$, and assume that they are simple. Then the left-hand side of
equation (3.3) can be written as

\bigskip

\begin{eqnarray*}
i\underset{n=1}{\overset{N}{\sum }}\frac{\Phi ^{(1)}(x,\mu _{n})}{%
S_{11}^{^{\prime }}(\mu _{n})}e^{i\mu _{n}\xi },\text{ }S_{11}^{^{\prime
}}(\mu _{n}) &=&\left. \frac{dS_{11}(\mu )}{d\mu }\right \vert _{\mu =\mu
_{n}}, \\
\Phi ^{(1)}(x,\mu _{n}) &=&C_{12}^{(n)}\Psi ^{(2)}(x,\mu
_{n})+C_{13}^{(n)}\Psi ^{(3)}(x,\mu _{n}).
\end{eqnarray*}

\bigskip

In addition, $\Psi ^{(2)}(x,\mu _{n})$ and $\Psi ^{(3)}(x,\mu _{n})$ must be
expressed using the triangular representation.

One should first replace $\Psi ^{(2)}$ and $\Psi ^{(3)}$ by their triangular
representations rearranging the right-hand side of equation (3.3). The
evaluation of the integrals in (3.3) then yields

\begin{equation}
-K^{(1)}(x,\xi )=\underset{p}{\sum }F_{1p}(x+\xi )h^{(p)}+\underset{p}{\sum }%
\underset{x}{\overset{\infty }{\int }}K^{(p)}(x,\theta )F_{1p}(\theta +\xi
)d\theta ,  \tag{3.4}
\end{equation}%
where

\begin{equation}
F_{1p}(\theta )=-i\underset{n=1}{\overset{N}{\sum }}\frac{C_{1p}^{(n)}}{%
S_{11}^{^{\prime }}(\mu _{n})}e^{i\mu _{n}\theta }+\frac{1}{2\pi }\overset{%
\infty }{\underset{-\infty }{\int }}\frac{S_{1p}(\mu )}{S_{11}(\mu )}e^{i\mu
\theta }d\mu ,\text{ }p=2,3,  \tag{3.5}
\end{equation}%
and $K^{(i)}(x,\theta )$ denote the kernels of the triangular
representations of the Jost functions $\Psi ^{(i)}.$

We now proceed to transform the remaining equations in (3.2),

\begin{equation*}
\Phi ^{(p)}(x,\mu )=S_{p1}(\mu )\Psi ^{(1)}(x,\mu )+\underset{\gamma }{\sum }%
S_{p\gamma }(\mu )\Psi ^{(\gamma )}(x,\mu ),\text{ }p,\gamma =2,3.
\end{equation*}

\bigskip

Let $W$ denote the matrix inverse to $S_{p\gamma },$ defined by $\underset{%
_{p^{\prime }}}{\sum }W_{pp^{\prime }}S_{p^{\prime }\gamma }=\delta
_{p\gamma }$. Then, by analogy with the previous equation, we obtain

\begin{equation*}
\underset{p^{\prime }}{\sum }W_{pp^{\prime }}\Phi ^{(p^{\prime })}(x,\mu )=%
\underset{p^{\prime }}{\sum }W_{pp^{\prime }}S_{p^{\prime }1}\Psi
^{(1)}(x,\mu )+\Psi ^{(p)}(x,\mu ).
\end{equation*}

\bigskip

To proceed with the transformation of this equation, it is necessary to
determine the relation between $W$ and $S_{11}(\mu ),$ $C_{12}$ and $C_{13}.$
Using the unitarity and unimodularity properties, it is straightforward to
show that

\bigskip

\begin{equation*}
W_{22}=\frac{S_{33}}{S_{11}^{\ast }},\text{ }W_{23}=-\frac{S_{23}}{%
S_{11}^{\ast }},\text{ }W_{32}=-\frac{S_{32}}{S_{11}^{\ast }},\text{ }W_{33}=%
\frac{S_{22}}{S_{11}^{\ast }},
\end{equation*}%
and also

\bigskip

\begin{equation*}
\underset{p^{\prime }=2}{\overset{3}{\sum }}W_{pp^{\prime }}S_{p^{\prime
}1}=-\frac{1}{S_{11}^{\ast }}\left \{ 
\begin{array}{c}
S_{12}^{\ast }\text{ for }p=2, \\ 
S_{13}^{\ast }\text{ for }p=3.%
\end{array}%
\right.
\end{equation*}

\bigskip

The equation can now be rewritten in analogy with (3.3) as

\begin{eqnarray}
&&\frac{1}{2\pi }\overset{\infty }{\underset{-\infty }{\int }}\left(
W_{pp^{\prime }}\Phi ^{(p^{\prime })}-h^{(p)}e^{i\mu x}\right) e^{-i\lambda
\mu \xi }d\mu  \TCItag{3.6} \\
&&\overset{\xi \geq x}{=}\frac{1}{2\pi }\overset{\infty }{\underset{-\infty }%
{\int }}\left( \Psi ^{(p)}-h^{(p)}e^{i\mu x}\right) e^{-i\lambda \xi }d\mu -%
\frac{1}{2\pi }\overset{\infty }{\underset{-\infty }{\int }}\frac{%
S_{1p}^{\ast }}{S_{11}^{\ast }}\Psi ^{(1)}e^{-i\mu \xi }d\mu .  \notag
\end{eqnarray}

The integrand on the left-hand side of equation (3.6) is analytic in the
region $\func{Im}\mu \leq 0,$ except at the points $\mu _{n}^{\ast },$ where
it has simple poles, as anticipated earlier.

To evaluate the residues, we assume that the potentials $q_{2,3}(x)$ vanish
outside a finite interval. Under the assumption, $\Phi ^{(i)},$ $\Psi ^{(i)}$
and $S_{ij}$ are analytic throughout the entire $\mu $. Hence, the residue
is given by%
\begin{equation*}
-[\frac{S_{1p}^{\ast }}{S_{11}^{\ast ^{\prime }}}\Psi ^{(1)}e^{-i\mu \xi
}]_{\mu =\mu _{n}}.
\end{equation*}

\bigskip Letting the cut-radius tend to infinity, we obtain that the residue
of the function $\left( W_{pp^{\prime }}\Phi ^{(p^{\prime })}-h^{(p)}e^{i\mu
x}\right) e^{-i\mu \xi }$ at the point $\mu _{n}^{\ast }$ given by

\begin{equation*}
\frac{C_{1p}^{\ast }}{[S_{11}^{^{\prime }}(\mu _{n})]^{\ast }}\Psi
^{(1)}(\mu _{n}^{\ast })e^{-i\mu _{n}^{\ast }\xi }.
\end{equation*}

Further integration of equation (3.4), together with the use of the
triangular representation, yields

\begin{equation}
K^{(p)}(x,\xi )=F_{p1}(x+\xi )h^{(1)}+\overset{\infty }{\underset{x}{\int }}%
K^{(1)}(x,\theta )F_{p1}(\theta +\xi )d\theta ,\text{ \ }F_{p1}=F_{1p}^{\ast
},\text{ \ }p=2,3.  \tag{3.7}
\end{equation}

\bigskip

The Gel'fand--Levitan--Marchenko equations (3.4) provide a means of
determining the kernels of the triangular representations of the Jost
functions from the scattering data.

Finally, to establish the relation between the kernels $K^{(i)}(x,\xi )$ of
the Jost functions and the potentials $q_{2,3}(x),$ it is necessary to
substitute the triangular representation into the original equations (3.1).
After some straightforward algebra, equations (3.1) yield

\begin{equation}
q_{2}(x)=-2K_{1}^{(2)}(x,x)=2K_{2}^{(1)\ast }(x,x),\text{ \ }%
q_{3}(x)=-2K_{1}^{(3)}(x,x)=2K_{3}^{(1)\ast }(x,x),  \tag{3.8}
\end{equation}%
where the lower index of $K_{j}^{(i)}$ numbers the elements $j$ of the
column vector $K^{(i)}$.

\section{\protect\bigskip\ Explicit Solutions in Inverse Scattering Method}

\textbf{4.1 Evolutionary GLM\ Equation:}

Before examining the spectral data of Manakov's problem as a function of $t$%
, we first derive the equations governing the evolution of the scattering
data for a general choice of the $3\times 3$ matrices $U$ and $V$ in
equations (2.1)-(2.2).

At each fixed value of $t$, the Jost functions associated with problem (2.1)
are defined by their asymptotic behavior

\bigskip

\begin{eqnarray*}
\Phi ^{(j)}(x,t) &\rightarrow &h^{(j)}\exp (i\mu _{j}^{(-)}x),\text{ \ }%
x\rightarrow -\infty , \\
\Psi ^{(j)}(x,t) &\rightarrow &h^{(j)}\exp (i\mu _{j}^{(+)}x),\text{ \ }%
x\rightarrow +\infty ,
\end{eqnarray*}

\bigskip

where

\bigskip

\begin{equation*}
h_{k}^{(j)}=\delta _{jk},\text{ \ }\underset{x\rightarrow \pm \infty }{\lim }%
(U)_{jk}=i\mu _{j}^{(\pm )}\delta _{jk},\text{ \ }j,k=1,2,3.
\end{equation*}

\bigskip

The dependence on the spectral parameter $\mu $ of $\Phi ^{(j)}$ and $\Psi
^{(j)}$ is suppressed for notational simplicity.

The transition matrix is defined in the standard way,

\bigskip

\begin{equation*}
\Phi ^{(i)}=\underset{k=1}{\overset{3}{\sum }}S_{ik}\Psi ^{(k)},
\end{equation*}%
and with $x\rightarrow +\infty $ differentiation of this equation results in

\bigskip

\begin{equation*}
\frac{\partial \Phi ^{(i)}}{\partial t}=\underset{k=1}{\overset{3}{\sum }}%
\frac{\partial S_{ik}}{\partial t}\Psi ^{(k)}.
\end{equation*}

\bigskip

At each fixed value of $t$, the Jost functions are related by linear
transformations to the functions $\widetilde{\Phi }^{(n)}(x,t)$ which are
solutions of equations (2.1)-(2.2),

\bigskip

\begin{equation*}
\Phi ^{(i)}(x,t)=\underset{k=1}{\overset{3}{\sum }}Q_{ik}(t)\widetilde{\Phi }%
^{(k)}(x,t),\text{ \ }\widetilde{\Phi }^{(k)}(x,t)=\underset{k=1}{\overset{3}%
{\sum }}P_{kn}(t)\Phi ^{(n)}(x,t).
\end{equation*}

\bigskip

At the initial time $t=0,$ these functions coincide with the Jost functions,

\bigskip

\begin{equation*}
\widetilde{\Phi }^{(k)}(x,t=0)=\Phi ^{(k)}(x,t=0).
\end{equation*}

\bigskip

Thereby

\bigskip

\begin{equation}
\frac{\partial \Phi ^{(i)}}{\partial t}=\underset{k=1}{\overset{3}{\sum }}(%
\frac{\partial Q_{ik}}{\partial t}\widetilde{\Phi }^{(k)}+Q_{ik}V\widetilde{%
\Phi }^{(k)})=\underset{k,n=1}{\overset{3}{\sum }}(\frac{\partial Q_{ik}}{%
\partial t}P_{kn}\Phi ^{(n)}+Q_{ik}P_{kn}V\Phi ^{(n)}).  \tag{4.1}
\end{equation}

\bigskip

Besides

\bigskip

\begin{equation*}
\frac{\partial \widetilde{\Phi }^{(i)}}{\partial t}=V\widetilde{\Phi }^{(i)}=%
\underset{n=1}{\overset{3}{\sum }}\frac{dP_{in}}{dt}\Phi ^{(n)}\text{ \ for }%
x\rightarrow -\infty ,
\end{equation*}

\bigskip

or

\bigskip

\begin{equation*}
\underset{n=1}{\overset{3}{\sum }}P_{in}V(-\infty )\Phi ^{(n)}=\underset{n=1}%
{\overset{3}{\sum }}\frac{dP_{in}}{dt}\Phi ^{(n)},\text{ \ }V(-\infty )=%
\underset{x\rightarrow -\infty }{\lim }V,\text{ \ }V_{ik}(-\infty
)=a_{i}(-\infty )\delta _{ik}.
\end{equation*}

\bigskip

Hence

\bigskip

\begin{equation*}
\frac{dP_{in}}{dt}=a_{n}(-\infty )P_{in},
\end{equation*}

\bigskip

which gives

\bigskip

\begin{equation*}
P_{in}=\delta _{in}\exp [a_{n}(-\infty )t],\text{ \ }Q_{nk}=\delta _{nk}\exp
[-a_{k}(-\infty )t].
\end{equation*}

\bigskip

Taking equation (4.1) and letting $x\rightarrow +\infty $,

\bigskip

\begin{equation*}
\frac{\partial \Phi ^{(i)}}{\partial t}=\underset{k=1}{\overset{3}{\sum }}%
\frac{dS_{ik}}{dt}\Psi ^{(k)}=\underset{i,k=1}{\overset{3}{\sum }}%
[-a_{i}(-\infty )S_{ik}\Psi ^{(k)}+S_{ik}a_{k}(+\infty )\Psi ^{(k)}],
\end{equation*}

\bigskip

we obtain the main equation governing the evolution of the scattering data,

\bigskip

\begin{equation*}
\frac{dS_{ik}}{dt}=S_{ik}[-a_{i}(-\infty )+a_{k}(+\infty )]
\end{equation*}

\bigskip

($\underset{x\rightarrow \pm \infty }{\lim V_{ik}}=a_{k}(\pm \infty )\delta
_{ik}$ is assumed).

Substituting equations (3.7) and (3.5) into equation (3.4) yields the
following time-evolution for the GLM:

\bigskip

\begin{theorem}
The second and third components $K_{p}^{(1)}(x,y),p=2,3$ of $K^{(1)}(x,y)$
are suitably obtained as%
\begin{equation}
-K_{p}^{(1)}(x,y)=\overset{N}{\underset{n=1}{\sum }}C_{p}^{(n)}e^{i\mu
_{n}(x+y)}+\underset{x}{\overset{\infty }{\int }}K_{p}^{(1)}(x,\tau
)\tciFourier (x,\tau ,y)d\tau ,  \tag{4.2}
\end{equation}%
\textit{where}

\begin{equation*}
\tciFourier (x,\tau ,y)=\underset{x}{\overset{\infty }{\int }}\left[ \overset%
{N}{\underset{n,m=1}{\sum }}\left( C_{2}^{(n)\ast
}C_{2}^{(m)}+C_{3}^{(n)\ast }C_{3}^{(m)}\right) e^{i(\mu _{m}-\mu _{n}^{\ast
})z-i\mu _{n}^{\ast }\tau +i\mu _{m}y}\right] dz
\end{equation*}

\textit{and}%
\begin{equation*}
C_{p}^{(n)}=-i[C_{1p}^{(n)}/S_{11}^{^{\prime }}(\mu _{n})]\exp [i\theta
_{n}t],p=2,3,
\end{equation*}%
\textit{where,} $\theta _{n}=-\frac{\gamma }{2\mu _{n}}$.
\end{theorem}

\textbf{4.2. N-Soliton Solutions:}

As given by equations (3.8) the solution of \ the CNKG equation (2.5) is

\begin{equation*}
q_{2}(x,t)=2K_{2}^{(1)\ast }(x,x;t),\text{ \ }q_{3}(x)=2K_{3}^{(1)\ast
}(x,x;t),
\end{equation*}

where $K_{p}^{(1)}(x,y),p=2,3$ satisfy the time-dependent GLM equation (4.2).

\begin{theorem}
\bigskip Let $q_{p}\left( x,t\right) ,$ $p=2,3$ be a coefficient of system
(3.1) satisfying equations (2.5 ), and let the system (3.1) have only a
discrete spectrum on the imaginary axis. Then the N-soliton solution of
(2.5) is obtained as a solution of the inverse problem in the form

\begin{equation}
\begin{array}{c}
q_{p}\left( x,t\right) =-2\underset{n=1}{\overset{N}{\sum }}C_{p}^{(n)\ast
}e^{-2i\mu _{n}^{\ast }(x)} \\ 
+2\underset{n,m=1}{\overset{N}{\sum }}\frac{i}{\mu _{m}^{\ast }-\mu _{n}}%
\left( C_{2}^{(n)}C_{2}^{(m)\ast }+C_{3}^{(n)}C_{3}^{(m)\ast }\right)
e^{-i\left( 2\mu _{m}^{\ast }-\mu _{n}\right) x}\chi _{n}^{\ast }\left(
x,t\right) ,\text{ }p=2,3,%
\end{array}
\tag{4.3}
\end{equation}

where

\begin{equation*}
C_{p}^{(n)}=-i[C_{1p}^{(n)}/S_{11}^{^{\prime }}(\mu _{n})]\exp [-i\frac{%
\gamma }{2\mu _{n}}t],\text{ }p=2,3,
\end{equation*}

and $\chi _{1},\chi _{2},...,\chi _{N}$ are the column entries of

\begin{equation*}
\chi =-(I-A)^{-1}d.
\end{equation*}

Here, $A$ is an $N\times N$ matrix whose entries are

\begin{equation*}
A_{kn}=\underset{m=1}{\overset{N}{\sum }}\frac{1}{(\mu _{m}-\mu _{n}^{\ast
})(\mu _{m}-\mu _{k}^{\ast })}(C_{2}^{(n)\ast }C_{2}^{(m)}+C_{3}^{(n)\ast
}C_{3}^{(m)})e^{i(2\mu _{m}-\mu _{n}^{\ast }-\mu _{k}^{\ast })x},\text{ }%
k,n=1,2,...,N,
\end{equation*}

and $d_{k}$ is the column vector with the entries

\begin{equation*}
d_{k}=\underset{n=1}{\overset{N}{\sum }}\frac{iC_{p}^{(n)}e^{i(2\mu _{n}-\mu
_{k}^{\ast })x}}{\mu _{n}-\mu _{k}^{\ast }},\text{ }k=1,2,...,N,
\end{equation*}

where $\mu _{n}$ are eigenvalues, and $C_{p}^{(n)},$ $p=2,3$ are the
normalization numbers of system (3.1).

\begin{proof}
Let $q_{p}\left( x,t\right) ,$ $p=2,3$ be coefficients of the (3.1) system,
and let these systems have only a discrete spectrum. For arbitrary fixed $x$
and $t$, equation (4.2) is an integral equation with a degenerate kernel

\begin{equation}
\begin{array}{c}
K_{p}^{(1)}(x,y;t)=-\underset{n=1}{\overset{N}{\sum }}C_{p}^{(n)}e^{i\mu
_{n}(x+y)} \\ 
-\underset{n,m=1}{\overset{N}{\sum }}\frac{i}{\mu _{m}-\mu _{n}^{\ast }}%
(C_{2}^{(n)\ast }C_{2}^{(m)}+C_{3}^{(n)\ast }C_{3}^{(m)})e^{i\mu _{m}y+i(\mu
_{m}-\mu _{n}^{\ast })x}\overset{\infty }{\underset{x}{\int }}%
K_{p}^{(1)}(x,y;t)e^{-i\mu _{n}^{\ast }y}dy.%
\end{array}
\tag{4.4}
\end{equation}

Here

\begin{equation*}
C_{p}^{(n)}=-i[C_{1p}^{(n)}/S_{11}^{^{\prime }}(\mu _{n})]\exp [-i\frac{%
\gamma }{2\mu _{n}}t],\text{ }p=2,3.
\end{equation*}

The solution of (14) has the form

\begin{equation*}
K_{p}^{(1)}(x,y;t)=-\underset{n=1}{\overset{N}{\sum }}C_{p}^{(n)}e^{i\mu
_{n}(x+y)}-\underset{n,m=1}{\overset{N}{\sum }}\frac{i}{\mu _{m}-\mu
_{n}^{\ast }}(C_{2}^{(n)\ast }C_{2}^{(m)}+C_{3}^{(n)\ast
}C_{3}^{(m)})e^{i\mu _{m}y+i(\mu _{m}-\mu _{n}^{\ast })x}\chi _{n}(x,t),
\end{equation*}

where $\chi _{n}(x,t)$ is chosen as

\begin{equation*}
\chi _{n}(x,t)=\overset{\infty }{\underset{x}{\int }}K_{p}^{(1)}(x,y;t)e^{-i%
\mu _{n}^{\ast }y}dy
\end{equation*}

and is a solution of the system of equations

\begin{equation*}
\begin{array}{c}
\chi _{k}(x,t)=\underset{n=1}{\overset{N}{\sum }}\frac{iC_{p}^{(n)}e^{i(2\mu
_{n}-\mu _{k}^{\ast })x}}{\mu _{n}-\mu _{k}^{\ast }} \\ 
+\underset{n,m=1}{\overset{N}{\sum }}\frac{1}{(\mu _{m}-\mu _{n}^{\ast
})(\mu _{m}-\mu _{k}^{\ast })}(C_{2}^{(n)\ast }C_{2}^{(m)}+C_{3}^{(n)\ast
}C_{3}^{(m)})e^{i(2\mu _{m}-\mu _{n}^{\ast }-\mu _{k}^{\ast })x}\chi
_{n}(x,t),\text{ }%
\end{array}%
k=1,2,...,N.
\end{equation*}

Introducing the matrix $A$ with the entries

\begin{equation*}
A_{kn}=\underset{m=1}{\overset{N}{\sum }}\frac{1}{(\mu _{m}-\mu _{n}^{\ast
})(\mu _{m}-\mu _{k}^{\ast })}(C_{2}^{(n)\ast }C_{2}^{(m)}+C_{3}^{(n)\ast
}C_{3}^{(m)})e^{i(2\mu _{m}-\mu _{n}^{\ast }-\mu _{k}^{\ast })x},\text{ }%
k,n=1,2,...,N
\end{equation*}

and the column vector $d_{k}$ with the entries

\begin{equation*}
d_{k}=\underset{n=1}{\overset{N}{\sum }}\frac{iC_{p}^{(n)}e^{i(2\mu _{n}-\mu
_{k}^{\ast })x}}{\mu _{n}-\mu _{k}^{\ast }},\text{ }k=1,2,...,N
\end{equation*}

and letting $\chi $ denote the columns $\chi _{1},\chi _{2},...,\chi _{N},$
we can rewrite this equation in the form

\begin{equation*}
(I-A)\chi =-d,
\end{equation*}

where $I$ is $N\times N$ identity matrix. Thus, we have

\begin{equation*}
\chi =-(I-A)^{-1}d.
\end{equation*}

Because $q_{p}\left( x,t\right) =2K_{p}^{(1)\ast }(x,x;t),$ $p=2,3,$ the
coefficient $q_{p}\left( x,t\right) $ is determined from reconstruction
formulas (3.8) as given in (4.3).
\end{proof}
\end{theorem}

\bigskip

\textbf{4.3. Examples:}

\textbf{a) One-Soliton Solutions}

\bigskip Let $N=1$, in the case of\textbf{\ }$\mu _{1}=i\sigma $ in (4.3),
we get

\begin{equation*}
e_{2}(x,t)=2K_{2}^{(1)\ast }(x,x;t),e_{3}(x,t)=2K_{3}^{(1)\ast }(x,x;t),%
\text{ }
\end{equation*}%
where

\begin{equation*}
K_{\beta }^{(1)}(x,y;t)=-\frac{C_{\beta }\exp (-\sigma (x+y))}{1+\frac{1}{%
4\sigma ^{2}}\exp (-4\sigma x)\left( \left\vert C_{2}\right\vert
^{2}+\left\vert C_{3}\right\vert ^{2}\right) },\text{ }\beta =2,3,
\end{equation*}%
and%
\begin{equation*}
C_{\beta }=-iC_{1\beta }/S_{11}^{^{\prime }}(\mu _{1})\exp [-i\frac{\gamma }{%
2\mu _{1}}t]=\omega _{1\beta }\exp [-\frac{\gamma }{2\sigma }t],\text{ }%
\beta =2,3.
\end{equation*}%
For the simplicity, it is denoted that $C_{1\beta }/S_{11}^{^{\prime }}(\mu
_{1})=i\omega _{1\beta }$ for $\beta =2,3.$

The following one solitons are obtained for the (2.5):

\bigskip

\begin{equation*}
e_{2}(x,t)=-\frac{2\omega _{12}e^{-2\sigma _{1}x-\frac{\gamma }{2\sigma }t}}{%
1+\frac{1}{4\sigma ^{2}}\left( \omega _{12}^{2}+\omega _{13}^{2}\right)
e^{-4\sigma x-\frac{\gamma }{\sigma }t}},\ \ e_{3}(x,t)=-\frac{2\omega
_{13}e^{-2\sigma x-\frac{\gamma }{2\sigma }t}}{1+\frac{1}{4\sigma ^{2}}%
\left( \omega _{12}^{2}+\omega _{13}^{2}\right) e^{-4\sigma x-\frac{\gamma }{%
\sigma }t}}.
\end{equation*}%
\textbf{b) Two-Soliton Solutions}

\bigskip Let $N=2$, in the case of $\mu _{1}=i\sigma _{1}$ and $\mu
_{2}=i\sigma _{2}$ in (4.3), we get

\begin{equation*}
q_{p}\left( x,t\right) =-2\underset{n=1}{\overset{2}{\sum }}C_{p}^{(n)\ast
}e^{-2\sigma _{n}x}-2\underset{n,m=1}{\overset{2}{\sum }}\frac{1}{\sigma
_{n}+\sigma _{m}}\left( C_{2}^{(n)}C_{2}^{(m)\ast
}+C_{3}^{(n)}C_{3}^{(m)\ast }\right) e^{-\left( \sigma _{n}+2\sigma
_{m}\right) x}\chi _{n}^{\ast }\left( x,t\right) ,\text{ }p=2,3,
\end{equation*}

where

\begin{eqnarray*}
\chi _{k}(x,t) &=&\frac{C_{p}^{(1)}e^{-(2\sigma _{1}+\sigma _{k})x}}{\sigma
_{1}+\sigma _{k}}+\frac{C_{p}^{(2)}e^{-(2\sigma _{2}+\sigma _{k})x}}{\sigma
_{2}+\sigma _{k}} \\
&&+\frac{1}{2\sigma _{1}(\sigma _{1}+\sigma _{k})}(C_{2}^{(1)\ast
}C_{2}^{(1)}+C_{3}^{(1)\ast }C_{3}^{(1)})e^{-(3\sigma _{1}+\sigma
_{k})x}\chi _{1}(x,t) \\
&&+\frac{1}{(\sigma _{1}+\sigma _{2})(\sigma _{2}+\sigma _{k})}%
(C_{2}^{(1)\ast }C_{2}^{(2)}+C_{3}^{(1)\ast }C_{3}^{(2)})e^{-(2\sigma
_{2}+\sigma _{1}+\sigma _{k})x}\chi _{1}(x,t) \\
&&+\frac{1}{(\sigma _{1}+\sigma _{2})(\sigma _{1}+\sigma _{k})}%
(C_{2}^{(2)\ast }C_{2}^{(1)}+C_{3}^{(2)\ast }C_{3}^{(1)})e^{-(2\sigma
_{1}+\sigma _{2}+\sigma _{k})x}\chi _{2}(x,t) \\
&&+\frac{1}{2\sigma _{2}(\sigma _{2}+\sigma _{k})}(C_{2}^{(2)\ast
}C_{2}^{(2)}+C_{3}^{(2)\ast }C_{3}^{(2)})e^{-(3\sigma _{2}+\sigma
_{k})x}\chi _{2}(x,t),\text{ }p=2,3;\text{ }k=1,2\text{,}
\end{eqnarray*}%
and

\begin{equation*}
C_{p}^{(n)}=-i[C_{1p}^{(n)}/S_{11}^{^{\prime }}(i\sigma _{n})]e^{-\frac{%
\gamma }{2\sigma _{n}}t}=\omega _{1p}^{(n)}e^{-\frac{\gamma }{2\sigma _{n}}%
t},\text{ }p=2,3;\text{ }n=1,2\text{. }
\end{equation*}

For the simplicity, it is denoted that $C_{1\beta }^{(n)}/S_{11}^{^{\prime
}}(i\sigma _{n})=i\omega _{1p}^{\left( n\right) }$ for $p=2,3$ and $n=1,2$.

Then the first component of two-soliton is

\begin{eqnarray*}
q_{2}\left( x,t\right) &=&-2(\omega _{12}^{(1)}e^{-2\sigma _{1}x-\frac{%
\gamma }{2\sigma _{1}}t}+\omega _{12}^{(2)}e^{-2\sigma _{2}x-\frac{\gamma }{%
2\sigma _{2}}t})-\frac{1}{\sigma _{1}}\left( (\omega
_{12}^{(1)})^{2}+(\omega _{13}^{(1)})^{2}\right) e^{-3\sigma _{1}x-\frac{%
\gamma }{\sigma _{1}}t}\chi _{1}^{\ast }\left( x,t\right) \\
&&-\frac{2}{\sigma _{1}+\sigma _{2}}\left( \omega _{12}^{(1)}\omega
_{12}^{(2)}+\omega _{13}^{(1)}\omega _{13}^{(2)}\right) e^{-\left( \sigma
_{1}+2\sigma _{2}\right) x-\frac{\gamma }{2\sigma _{1}}t-\frac{\gamma }{%
2\sigma _{2}}t}\chi _{1}^{\ast }\left( x,t\right) \\
&&-\frac{2}{\sigma _{1}+\sigma _{2}}\left( \omega _{12}^{(1)}\omega
_{12}^{(2)}+\omega _{13}^{(1)}\omega _{13}^{(2)}\right) e^{-\left( \sigma
_{2}+2\sigma _{1}\right) x-\frac{\gamma }{2\sigma _{1}}t-\frac{\gamma }{%
2\sigma _{2}}t}\chi _{2}^{\ast }\left( x,t\right) \\
&&-\frac{1}{\sigma _{2}}\left( (\omega _{12}^{(2)})^{2}+(\omega
_{13}^{(2)})^{2}\right) e^{-3\sigma _{2}x-\frac{\gamma }{\sigma _{2}}t}\chi
_{2}^{\ast }\left( x,t\right) ,
\end{eqnarray*}

\begin{equation*}
\chi _{1}\left( x,t\right) =\frac{AF+CD}{BF+CE},\text{ \ \ }\chi _{2}\left(
x,t\right) =\frac{AE+BD}{BF+CE},
\end{equation*}%
\begin{equation*}
A=\frac{\omega _{12}^{(1)}e^{-3\sigma _{1}x-\frac{\gamma }{2\sigma _{1}}t}}{%
2\sigma _{1}}+\frac{\omega _{12}^{(2)}e^{-\left( \sigma _{1}+2\sigma
_{2}\right) x-\frac{\gamma }{2\sigma _{2}}t}}{\sigma _{1}+\sigma _{2}},
\end{equation*}

\begin{equation*}
B=-\frac{\left( (\omega _{12}^{(1)})^{2}+(\omega _{13}^{(1)})^{2}\right)
e^{-4\sigma _{1}x-\frac{\gamma }{\sigma _{1}}t}}{4\sigma _{1}^{2}}-\frac{%
\left( \omega _{12}^{(1)}\omega _{12}^{(2)}+\omega _{13}^{(1)}\omega
_{13}^{(2)}\right) e^{-2(\sigma _{1}+\sigma _{2})x-\frac{\gamma }{2\sigma
_{1}}t-\frac{\gamma }{2\sigma _{2}}t}}{(\sigma _{1}+\sigma _{2})^{2}},
\end{equation*}

\begin{equation*}
C=-\frac{\left( \omega _{12}^{(1)}\omega _{12}^{(2)}+\omega
_{13}^{(1)}\omega _{13}^{(2)}\right) e^{-\left( 3\sigma _{1}+\sigma
_{2}\right) x-\frac{\gamma }{2\sigma _{1}}t-\frac{\gamma }{2\sigma _{2}}t}}{%
2\sigma _{1}(\sigma _{1}+\sigma _{2})}-\frac{\left( (\omega
_{12}^{(2)})^{2}+(\omega _{13}^{(2)})^{2}\right) e^{-\left( 3\sigma
_{2}+\sigma _{1}\right) x-\frac{\gamma }{\sigma _{2}}t}}{2\sigma _{2}(\sigma
_{1}+\sigma _{2})},
\end{equation*}

\begin{equation*}
D=\frac{\omega _{12}^{(1)}e^{-\left( 2\sigma _{1}+\sigma _{2}\right) x-\frac{%
\gamma }{2\sigma _{1}}t}}{\sigma _{1}+\sigma _{2}}+\frac{\omega
_{12}^{(2)}e^{-3\sigma _{2}x-\frac{\gamma }{2\sigma _{2}}t}}{2\sigma _{2}},
\end{equation*}

\begin{equation*}
E=-\frac{\left( (\omega _{12}^{(1)})^{2}+(\omega _{13}^{(1)})^{2}\right)
e^{-\left( 3\sigma _{1}+\sigma _{2}\right) x-\frac{\gamma }{\sigma _{1}}t}}{%
2\sigma _{1}(\sigma _{1}+\sigma _{2})}-\frac{\left( \omega _{12}^{(1)}\omega
_{12}^{(2)}+\omega _{13}^{(1)}\omega _{13}^{(2)}\right) e^{-\left( 3\sigma
_{2}+\sigma _{1}\right) x-\frac{\gamma }{2\sigma _{1}}t-\frac{\gamma }{%
2\sigma _{2}}t}}{2\sigma _{2}(\sigma _{1}+\sigma _{2})},
\end{equation*}

\begin{equation*}
F=-\frac{\left( \omega _{12}^{(1)}\omega _{12}^{(2)}+\omega
_{13}^{(1)}\omega _{13}^{(2)}\right) e^{-2(\sigma _{1}+\sigma _{2})x-\frac{%
\gamma }{2\sigma _{1}}t-\frac{\gamma }{2\sigma _{2}}t}}{(\sigma _{1}+\sigma
_{2})^{2}}-\frac{\left( (\omega _{12}^{(2)})^{2}+(\omega
_{13}^{(2)})^{2}\right) e^{-4\sigma _{2}x-\frac{\gamma }{\sigma _{2}}t}}{%
4\sigma _{2}^{2}}.
\end{equation*}

The second component of two soliton is

\begin{eqnarray*}
q_{3}\left( x,t\right) &=&-2(\omega _{13}^{(1)}e^{-2\sigma _{1}x-\frac{%
\gamma }{2\sigma _{1}}t}+\omega _{13}^{(2)}e^{-2\sigma _{2}x-\frac{\gamma }{%
2\sigma _{2}}t})-\frac{1}{\sigma _{1}}\left( (\omega
_{12}^{(1)})^{2}+(\omega _{13}^{(1)})^{2}\right) e^{-3\sigma _{1}x-\frac{%
\gamma }{\sigma _{1}}t}\chi _{1}^{\ast }\left( x,t\right) \\
&&-\frac{2}{\sigma _{1}+\sigma _{2}}\left( \omega _{12}^{(1)}\omega
_{12}^{(2)}+\omega _{13}^{(1)}\omega _{13}^{(2)}\right) e^{-\left( \sigma
_{1}+2\sigma _{2}\right) x-\frac{\gamma }{2\sigma _{1}}t-\frac{\gamma }{%
2\sigma _{2}}t}\chi _{1}^{\ast }\left( x,t\right) \\
&&-\frac{2}{\sigma _{1}+\sigma _{2}}\left( \omega _{12}^{(1)}\omega
_{12}^{(2)}+\omega _{13}^{(1)}\omega _{13}^{(2)}\right) e^{-\left( \sigma
_{2}+2\sigma _{1}\right) x-\frac{\gamma }{2\sigma _{1}}t-\frac{\gamma }{%
2\sigma _{2}}t}\chi _{2}^{\ast }\left( x,t\right) \\
&&-\frac{1}{\sigma _{2}}\left( (\omega _{12}^{(2)})^{2}+(\omega
_{13}^{(2)})^{2}\right) e^{-3\sigma _{2}x-\frac{\gamma }{\sigma _{2}}t}\chi
_{2}^{\ast }\left( x,t\right) ,
\end{eqnarray*}

\begin{equation*}
\chi _{1}\left( x,t\right) =\frac{KF+CL}{BF+CE},\text{ \ \ }\chi _{2}\left(
x,t\right) =\frac{KE+BL}{BF+CE},
\end{equation*}

\begin{equation*}
K=\frac{\omega _{13}^{(1)}e^{-3\sigma _{1}x-\frac{\gamma }{2\sigma _{1}}t}}{%
2\sigma _{1}}+\frac{\omega _{13}^{(2)}e^{-\left( \sigma _{1}+2\sigma
_{2}\right) x-\frac{\gamma }{2\sigma _{2}}t}}{\sigma _{1}+\sigma _{2}},\text{
}L=\frac{\omega _{13}^{(1)}e^{-\left( 2\sigma _{1}+\sigma _{2}\right) x-%
\frac{\gamma }{2\sigma _{1}}t}}{\sigma _{1}+\sigma _{2}}+\frac{\omega
_{13}^{(2)}e^{-3\sigma _{2}x-\frac{\gamma }{2\sigma _{2}}t}}{2\sigma _{2}}.
\end{equation*}

\section{Conclusion}

The objective of the present work is to construct exact N-soliton solutions
for the coupled negative first order Klein-Gordon (CNKG) equation via the
inverse scattering method (ISM). By examining the properties of the spectral
functions through the associated Volterra integral equations, the
Gelfand--Levitan--Marchenko\ (GLM) equations are obtained. A direct
correspondence between the potential functions appearing in the
zero-curvature representation and the original coupled equation is then
established. Solving the resulting GLM equations leads to explicit formulae
for the N-soliton solutions. This formulation may serve as a foundation for
future investigations of other, more general coupled PDEs of negative order.


\bigskip

\begin{thebibliography}{99}
	
	\bibitem{TsuchidaWadati1998}
	T. Tsuchida and M. Wadati,
	The coupled modified Korteweg--de Vries equations,
	J. Phys. Soc. Jpn. \textbf{67} (1998), 1175--1187.
	\newblock DOI: 10.1143/JPSJ.67.1175.
	
	\bibitem{IwaoHirota1997}
	M. Iwao and R. Hirota,
	Soliton solutions of a coupled modified KdV equations,
	J. Phys. Soc. Jpn. \textbf{66} (1997), 577--588.
	\newblock DOI: 10.1143/JPSJ.66.577.
	
	\bibitem{IsmailovSabaz2025}
	M. I. Ismailov and C. Sabaz,
	Inverse scattering method via the Gel'fand--Levitan--Marchenko equation
	for some negative-order nonlinear wave equations,
	Theor. Math. Phys. \textbf{222} (2025), 20--33.
	\newblock DOI: 10.1134/S0040577925010039.
	
	\bibitem{IsmailovSabaz2023}
	M. I. Ismailov and C. Sabaz,
	Inverse scattering method via Riemann--Hilbert problem for nonlinear
	Klein--Gordon equation coupled with a scalar field,
	J. Phys. Soc. Jpn. \textbf{92} (2023), 104001.
	\newblock DOI: 10.7566/JPSJ.92.104001.
	
	\bibitem{HirotaOhta1991}
	R. Hirota and Y. Ohta,
	Hierarchies of coupled soliton equations. I,
	J. Phys. Soc. Jpn. \textbf{60} (1991), 798--809.
	\newblock DOI: 10.1143/JPSJ.60.798.
	
	\bibitem{NovikovEtAl1984}
	S. P. Novikov, S. V. Manakov, L. P. Pitaevskii, and V. E. Zakharov,
	\textit{Theory of Solitons: The Inverse Scattering Method},
	Consultants Bureau, New York, 1984.
	
	\bibitem{AblowitzSegur1981}
	M. J. Ablowitz and H. Segur,
	\textit{Solitons and the Inverse Scattering Transform},
	SIAM Studies in Applied Mathematics, Vol.~4,
	SIAM, Philadelphia, 1981.
	\newblock DOI: 10.1137/1.9781611970883.
	
	\bibitem{LiuGuoZhang2025}
	Y.-H. Liu, R. Guo, and J.-W. Zhang,
	Inverse scattering transform for the discrete nonlocal PT-symmetric
	nonlinear Schr\"odinger equation with nonzero boundary conditions,
	Physica D \textbf{472} (2025), 134528.
	\newblock DOI: 10.1016/j.physd.2025.134528.
	
	\bibitem{LiLiHuang2026}
	Y. Li, T. Li, and L. Huang,
	Inverse scattering transform for the coupled Yajima--Oikawa systems,
	Physica D \textbf{486} (2026), 135054.
	\newblock DOI: 10.1016/j.physd.2025.135054.
	
	\bibitem{MaimistovEtAl1990}
	A. I. Maimistov, A. M. Basharov, S. O. Elyutin, and Yu. M. Sklyarov,
	Present state of self-induced transparency theory,
	Phys. Rep. \textbf{191} (1990), 1--108.
	\newblock DOI: 10.1016/0370-1573(90)90142-O.
	
	\bibitem{Ma2018}
	W. X. Ma,
	Riemann--Hilbert problems and $N$-soliton solutions for a coupled
	mKdV system,
	J. Geom. Phys. \textbf{132} (2018), 45--54.
	
	\bibitem{Manakov1974}
	S. V. Manakov,
	On the theory of two-dimensional stationary self-focusing of
	electromagnetic waves,
	Sov. Phys. JETP \textbf{38} (1974), 248--253.
	
\end{thebibliography}
\end{document}